\documentclass[runningheads]{llncs}
\usepackage[T1]{fontenc}
\usepackage{microtype} 
\usepackage{ellipsis} 
\usepackage[abbreviations,british]{foreign} 

\usepackage{synttree, bussproofs, 
 rotating, xcolor, mathrsfs}
\usepackage{esvect, pgf, tikz, color}
\usetikzlibrary{arrows, automata}
\usepackage{algorithm2e}
\usepackage[all]{xy}
\usepackage{wrapfig}
\usepackage{algorithmicx}
\usepackage{algpseudocode}
\usepackage{scalerel,stackengine}
\usepackage{changepage, enumerate, proof, relsize,trfsigns,
comment}
\usepackage{multicol}
\usepackage[normalem]{ulem}
\usepackage{upgreek}
\usepackage{amsfonts}
\usepackage{hyperref}
\usepackage{float}
\usepackage{amsmath}

\newcommand{\sol}{\lambda}
\newcommand{\langx}{\mathscr{L}_{x}}

\newcommand{\biset}{\Uppi}
\newcommand{\diset}{\Upsigma}

\newcommand{\notp}{\neg p} 
\newcommand{\test}{?}

\newcommand{\pdlbox}[1]{[#1]}
\newcommand{\pdldia}[1]{\langle #1 \rangle}
\newcommand{\iffi}{\textit{iff} }

\newcommand{\pra}{\alpha}
\newcommand{\prb}{\beta}
\newcommand{\prc}{\gamma}
\newcommand{\choice}{{\mathbin{\scaleobj{.75}{\cup}}}}
\newcommand{\prastar}{\alpha^{*}}

\renewcommand{\phi}{\varphi}

\newcommand{\pdl}{\mathsf{PDL}}

\newcommand{\pdlvdash}{\models}
\newcommand{\imp}{\rightarrow}
\newcommand{\negnnf}[1]{\overline{#1}} 

\SetKwFor{Loop}{Loop}{}{EndLoop}

\newcommand{\propset}{\mathsf{Prop}}

\newcommand{\progset}{\mathsf{Prog}}

\newcommand{\lang}{\mathscr{L}}
\newcommand{\langpdl}{\mathscr{L}_{\pdl}}

\newcommand{\var}{\mathsf{Var}}

\newcommand{\calr}{\mathcal{R}}

\newcommand{\dfn}{Definition} 
\newcommand{\thm}{Theorem} 
\newcommand{\fig}{Figure} 
\newcommand{\ifandonlyif}{iff } 

\usepackage{graphicx}
\usepackage{color}

\begin{document}
\title{On Explicit Solutions to Fixed-Point Equations in Propositional Dynamic Logic} 
\titlerunning{On Explicit Solutions to Fixed-Point Equations in PDL} 
%
\author{Tim S. Lyon\inst{1}\orcidID{0000-0003-3214-0828}}
\authorrunning{T.S. Lyon}
%
\institute{Technische Universit\"at Dresden, N\"othnitzer Str. 46, 01069 Dresden, Germany
\email{timothy\_stephen.lyon@tu-dresden.de}\\
\url{https://sites.google.com/view/timlyon}}
\maketitle              
\begin{abstract}
Propositional dynamic logic ($\pdl$) is an important modal logic used to specify and reason about the behavior of software. A challenging problem in the context of $\pdl$ is solving fixed-point equations, i.e., formulae of the form $x \equiv \phi(x)$ such that $x$ is a propositional variable and $\phi(x)$ is a formula containing $x$. A solution to such an equation is a formula $\psi$ that omits $x$ and satisfies $\psi \equiv \phi(\psi)$, where $\phi(\psi)$ is obtained by replacing all occurrences of $x$ with $\psi$ in $\phi(x)$. In this paper, we identify a novel class of $\pdl$ formulae arranged in two dual hierarchies for which every fixed-point equation $x \equiv \phi(x)$ has a solution. Moreover, we not only prove the existence of solutions for all such equations, but also provide an explicit solution $\psi$ for each fixed-point equation.

\end{abstract}
%
%
%

\section{Introduction}\label{sec:intro}

Propositional dynamic logic ($\pdl$) was introduced by Fischer and Ladner~\cite{FisLad79} as a modal logic for modeling program execution. A characteristic feature of the logic is the inclusion of modalities of the form $\pdlbox{\pra}$ and $\pdldia{\pra}$ such that $\pra$ is taken to be a computer program in a programming language. A formula of the form $\pdlbox{\pra}\phi$ is read as ‘the proposition $\phi$ holds in all states reachable by executing the program $\pra$' and a formula of the form $\pdldia{\pra}\phi$ is read as ‘the proposition $\phi$ holds in some state reachable by executing the program $\pra$.'

$\pdl$ is widely regarded for its expressivity, decidability, and capacity for modeling dynamic systems. The logic's high degree of expressive power is due to the inclusion of fixed-point constructs, viz., the modalities $\pdldia{\prastar}$ and $\pdlbox{\prastar}$, which permit the specification of (co)inductive definitions. Such operators are essential in program specification and verification for describing behaviors over potentially infinite structures such as safety properties (e.g., ‘no bad state will be reached’), liveness properties (e.g., ‘a good state will eventually be reached’), and termination (e.g., ‘every execution eventually ends’); cf.~\cite{BraSti07}.


In the context of $\pdl$, finding explicit solutions for \emph{fixed-point equations} has proven to be a challenging problem. A fixed-point equation is a formula of the form $x \equiv \phi(x)$ such that $x$ is a propositional variable and $\phi(x)$ is a formula containing $x$. A \emph{solution} to such an equation is a formula $\psi$ not containing $x$ and satisfying $\psi \equiv \phi(\psi)$, where $\phi(\psi)$ is obtained by replacing every occurrence of $x$ in $\phi$ with $\psi$. Aside from being a problem of intrinsic theoretical interest, solving fixed-point equations has relevance to the longstanding open problem of if $\pdl$ enjoys Craig interpolation~\cite{BorGat20}.\footnote{In fact, the observations made in this paper grew out of the author's investigation of the Craig interpolation property for $\pdl$.} As a case in point, a proof-theoretic methodology relying on cyclic sequent systems has been used in recent years to establish interpolation for certain modal fixed-point logics~\cite{AfsLei22,Sha14}. A critical aspect of this methodology is the solution of fixed-point equations in the construction of interpolants. Naturally, this demonstrates the value of solving such equations in the context of computing interpolants and establishing interpolation properties. 
Moreover, this also has practical relevance as 
Craig interpolation can be used to speed-up model-checking algorithms~\cite{McM18}.


To best of the author's knowledge, there appears to be little to no work on solving fixed-point equations for $\pdl$. Therefore, the aim of the current article is to take a first step toward rectifying this gap in the literature. This paper makes two contributions: first, we identify a non-trivial class of $\pdl$ formulae (arranged in two dual hierarchies of increasing formulaic complexity) such that for every formula $\phi(x)$ in the class, $x \equiv \phi(x)$ has a solution. Second, we not only prove the existence of solutions for all such equations, but provide an explicit witness~$\psi$ as a solution for each equation. We note that although it is straightforward to \emph{verify} the solutions to such equations (see \thm~\ref{thm:fixed-point-solution}), the novelty of the present work consists of having identified the class of solvable equations in the first place.\\




\section{Preliminaries: Propositional Dynamic Logic}\label{sec:prelims}

We let $\propset := \{p,q,r,\ldots\}$ be a countable set of \emph{propositional atoms}, $\var := \{x, y, z, \ldots\}$ be a countable set of \emph{variables}, and $\progset := \{a, b, c, \ldots\}$ be a countable set of \emph{atomic programs}. We define the language $\lang$ to be the set of formulae generated by the grammar in BNF shown below left and let a \emph{program} be an expression generated from the grammar shown below right:
$$
\phi ::= x \ | \ p \ | \ \notp \ | \ \phi \lor \phi \ | \ \phi \land \phi \ | \ \pdldia{\pra} \phi \ | \ \pdlbox{\pra} \phi
\qquad
\pra ::= a \ | \ \phi \test \ | \ \pra ; \pra \ | \ \pra \choice \pra \ | \ \pra^{*}
$$
where $x$ ranges over $\var$, $p$ ranges over $\propset$, and $a$ ranges over $\progset$. We use $\phi$, $\psi$, $\chi$, etc. to denote formulae in $\lang$ and $\pra$, $\prb$, $\prc$, etc. to denote programs. 

Although propositional dynamic logic does not traditionally include variables in formulae, we have opted to include them as they are helpful in formulating fixed-point equations. We let $\var(\phi)$ and $\var(\pra)$ denote the set of variables occurring in~$\phi$ and $\pra$, respectively. We let $\langpdl := \{\phi \in \lang \ | \ \var(\phi) = \emptyset\}$ be the traditional language of propositional dynamic logic and let $\langx := \{\phi \in \lang \ | \ x \not \in \var(\phi)\}$ be the set of all $x$-free formulae in $\lang$. We remark that variables can be considered positive propositional atoms and will be given the same interpretation (see \dfn~\ref{def:semantics-pdl} below). If a formula $\phi$ contains an occurrence of the variable $x$, then we may indicate this by writing $\phi$ as $\phi(x)$.

We define the negation $\negnnf{\phi}$ of a formula $\phi$ as follows: for $p \in \propset$, $\negnnf{p} := \notp$ and $\negnnf{\notp} := p$, and for $x \in \var$, $\negnnf{x} := x$. Furthermore, $\negnnf{\phi \lor \psi} := \negnnf{\phi} \land \negnnf{\psi}$, $\negnnf{\phi \land \psi} := \negnnf{\phi} \lor \negnnf{\psi}$, $\negnnf{\pdldia{\pra} \phi} := \pdlbox{\pra} \negnnf{\phi}$, and $\negnnf{\pdlbox{\pra} \phi} := \pdldia{\pra} \negnnf{\phi}$. Observe that the negation operation has no effect on a variable; defining negation in this way will be helpful in the next section. Last, we define $\phi \imp \psi := \negnnf{\phi} \lor \psi$, $\phi \equiv \psi := (\psi \imp \phi) \land (\psi \imp \phi)$, and for a fixed $p \in \propset$, we let $\top := p \lor \notp$ and $\bot := p \land \notp$.

\begin{definition}[Semantics]\label{def:semantics-pdl} A \emph{model} is a tuple $M = (W,\calr,V)$ such that
(1) $W$ is a non-empty set of worlds $w$, $u$, $v$, etc., (2) $\calr := \{R_{\pra} \subseteq W \times W \ | \ \pra \in \progset\}$ is a family of binary relations over $W$, each indexed with an atomic program, and (3) $V \colon (\propset \cup \var) \to 2^{W}$ is a \emph{valuation function}. We define relations for complex programs below, where $\circ$ stands for the usual composition operation on relations and $(R_{\pra})^{*}$ is the reflexive-transitive closure of $R_{\pra}$.
\begin{itemize}

\item $R_{\phi \test} := \{(w,w) \ | \ M, w \models \phi \}$;

\item $R_{\pra;\prb} := R_{\pra} \! \circ R_{\prb}$;

\item $R_{\pra \choice \prb} := R_{\pra} \! \cup R_{\prb}$;

\item $R_{\pra^{*}} := (R_{\pra})^{*}$.

\end{itemize}
\noindent
We recursively define the \emph{satisfaction} of a formula $\phi$ in a model $M$ at world $w$, written $M,w \pdlvdash \phi$, as follows:
\begin{itemize}

\item $M,w \pdlvdash A$ \ifandonlyif $w \in V(A)$, for $A \in \propset \cup \var$;

\item $M,w \pdlvdash \notp$ \ifandonlyif $w \not\in V(p)$;

\item $M,w \pdlvdash \phi \lor \psi$ \ifandonlyif $M,w \pdlvdash \phi$ or $M,w \pdlvdash \psi$;

\item $M,w \pdlvdash \phi \land \psi$ \ifandonlyif $M,w \pdlvdash \phi$ and $M,w \pdlvdash \psi$;

\item $M,w \pdlvdash \pdldia{\pra} \phi$ \ifandonlyif $\exists u \in R_{\pra}(w)$, $M,u \pdlvdash \phi$;

\item $M,w \pdlvdash \pdlbox{\pra} \phi$ \ifandonlyif $\forall u \in R_{\pra}(w)$, $M,u \pdlvdash \phi$.


\end{itemize}
We write $ \pdlvdash \phi$ and say that $\phi$ is \emph{valid} \iffi for every model $M = (W,\calr,V)$ and all $w \in W$, $M, w \pdlvdash \phi$. 
We define \emph{propositional dynamic logic} ($\pdl$) to be the set $\pdl \subseteq \langpdl$ of all valid formulae in $\langpdl$.
\end{definition}

It is well known that the logic $\pdl$ admits a finite axiomatization~\cite{KozPar81,Seg77}. We will not review the axioms of $\pdl$ here, however, we list a set of useful equivalences in \fig~\ref{fig:equivs} that are valid in $\pdl$ and which are used in the subsequent section to solve $\pdl$ fixed-point equations. Last, we note that the interested reader can consult Fischer and Ladner~\cite{FisLad79} for an introduction to $\pdl$. 

\begin{figure}[H]
\centering

\begin{multicols}{2}
\begin{description}

\item[E1] $\pdlbox{\pra} \pdlbox{\prb} \phi \equiv \pdlbox{\pra ; \prb} \phi$

\item[E2] $\phi \land \psi \equiv \pdldia{\phi\test}\psi$

\item[E3] $\pdlbox{\pra}(\phi \land \psi) \equiv \pdlbox{\pra} \phi \land \pdlbox{\pra} \psi$

\item[E4] $\pdlbox{\prastar} \phi \equiv \phi \land \pdlbox{\pra} \pdlbox{\prastar} \phi$

\item[E5] $\pdlbox{\top \test} \phi \equiv \phi$

\item[E6] $\pdldia{\pra} \pdldia{\prb} \phi \equiv \pdldia{\pra ; \prb} \phi$

\item[E7] $\phi \lor \psi \equiv \pdlbox{\negnnf{\phi}\test}\psi$

\item[E8] $\pdldia{\pra}(\phi \lor \psi) \equiv \pdldia{\pra} \phi \lor \pdldia{\pra} \psi$

\item[E9] $\pdldia{\prastar} \phi \equiv \phi \lor \pdldia{\pra} \pdldia{\prastar} \phi$

\item[E10] $\pdldia{\bot \test} \phi \equiv \phi$

\end{description}
\end{multicols}

\caption{Useful $\pdl$ equivalences.\label{fig:equivs}}
\end{figure}

\section{Solving PDL Fixed-Point Equations}\label{sec:solutions}

A \emph{fixed-point equation} is defined to be a formula of the form $x \equiv \phi(x)$ such that $x \in \var$ and $\phi(x) \in \lang$. Given formulae $\psi, \phi(x) \in \lang$, we let $\phi(\psi)$ denote the formula obtained by replacing $\psi$ for each occurrence of $x$ in $\phi(x)$. We define a \emph{solution} to a fixed-point equation $x \equiv \phi(x)$ to be a formula $\psi \in \langx$ such that $\psi \equiv \phi(\psi)$ is valid. In this section, we show that if $\phi(x)$ is in one of two forms, dubbed \emph{$\biset^{x}$-form} and \emph{$\diset^{x}$-form}, then $x \equiv \phi(x)$ has a solution. 


Let us fix a variable $x \in \var$. We define a hierarchy of formulae in $\lang$ of increasing complexity relative to $x$. The base level of the hierarchy is the set~$\biset_{0}^{x}$. Above the base level, odd levels $\biset_{2n+1}^{x}$ and even levels $\biset_{2n+2}^{x}$ are defined as shown below with $n \in \mathbb{N}$. The set $\biset^{x}$ is obtained by taking the union of all~levels. 
$$
\biset_{0}^{x} := \{\phi \lor (\psi \land x) \ | \ \phi, \psi \in \langx \}
\quad
\biset_{2n+1}^{x} := \{\pdlbox{\pra} \phi \ | \ \phi \in \biset_{2n}^{x} \ and \ x \not\in \var(\pra) \}
$$
$$
\biset_{2n+2}^{x} := \{\phi \lor (\psi \land \chi) \ | \ \chi \in \biset_{2n+1}^{x} \ and \ \phi, \psi \in \langx \}
\quad
\biset^{x} := \bigcup_{n \in \mathbb{N}} \biset_{n}^{x}
$$
We also define a dual hierarchy obtained by negating the formulae occurring in all sets defined above. In particular, we let $\diset_{n}^{x} := \{\negnnf{\phi} \ | \ \phi \in \biset_{n}^{x}\}$ for $n \in \mathbb{N}$ and $\diset^{x} := \{\negnnf{\phi} \ | \ \phi \in \biset^{x}\}$. We say that a formula $\phi(x) \in \lang$ is in \emph{$\biset^{x}$-form} or \emph{$\diset^{x}$-form} \iffi $\phi(x) \in \biset^{x}$ or $\phi(x) \in \diset^{x}$, respectively.

Let us discuss the shape of formulae in $\biset^{x}$-form and $\diset^{x}$-form. Each formula $\phi(x) \in \biset^{x}_{2n+1}$ and $\psi(x) \in \diset^{x}_{2n+1}$ is of form (\ref{eq:odd-form-1}) and (\ref{eq:odd-form-3}), respectively, which are shown below. We have also re-written forms (\ref{eq:odd-form-1}) and (\ref{eq:odd-form-3}) into their respectively equivalent forms (\ref{eq:odd-form-2}) and (\ref{eq:odd-form-4}) as these forms may be easier for the reader to parse and are helpful in the proof of \thm~\ref{thm:fixed-point-solution} below. Form (\ref{eq:odd-form-2}) may be obtained from (\ref{eq:odd-form-1}) by using the equivalences E1, E2, and E7 whereas (\ref{eq:odd-form-4}) may be obtained from (\ref{eq:odd-form-3}) by using the equivalences E2, E6, and E7. (NB. See \fig~\ref{fig:equivs} for the list of $\pdl$ equivalences E1--E10.) Note that if $\phi(x) \in \biset^{x}_{2n}$ or $\psi(x) \in \diset^{x}_{2n}$, then $\phi(x)$ and $\psi(x)$ have almost the same shape as formulae (\ref{eq:odd-form-1}) and (\ref{eq:odd-form-3}), but without the initial $\pdlbox{\pra_{1}}$ and $\pdldia{\pra_{1}}$ modalities, respectively. Therefore, $\phi(x) \in \biset^{x}_{2n}$ and $\psi(x) \in \diset^{x}_{2n}$ are respectively equivalent to formulae of forms (\ref{eq:odd-form-2}) and (\ref{eq:odd-form-4}), but where $\pdlbox{\pra_{1}; \negnnf{\phi}_{1} \test}$ is replaced by $\pdlbox{\negnnf{\phi}_{1} \test}$ and $\pdldia{\pra_{1}; \phi_{1} \test}$ is replaced by $\pdldia{\phi_{1} \test}$, respectively.
\begin{eqnarray}
\phi(x) & = & \pdlbox{\pra_{1}} \Big( \phi_{1} \lor \Big( \psi_{1} \land \pdlbox{\pra_{2}}\big( \phi_{2} \lor \big( \psi_{2} \land \ldots \pdlbox{\pra_{n}} (\phi_{n} \lor (\psi_{n} \land x)) \ldots \big)\big)\Big)\Big)\label{eq:odd-form-1}\\
& \equiv & \pdlbox{\pra_{1}; \negnnf{\phi}_{1} \test} \pdldia{\psi_{1} \test} \pdlbox{\pra_{2}; \negnnf{\phi}_{2} \test} \pdldia{\psi_{2} \test} \cdots \pdlbox{\pra_{n}; \negnnf{\phi}_{n} \test} \pdldia{\psi_{n} \test}x\label{eq:odd-form-2}\\
\psi(x) & = & \pdldia{\pra_{1}} \Big( \phi_{1} \land \Big( \psi_{1} \lor \pdldia{\pra_{2}}\big( \phi_{2} \land \big( \psi_{2} \lor \ldots \pdldia{\pra_{n}} (\phi_{n} \land (\psi_{n} \lor x)) \ldots \big)\big)\Big)\Big)\label{eq:odd-form-3}\\
& \equiv & \pdldia{\pra_{1}; \phi_{1} \test} \pdlbox{\negnnf{\psi}_{1} \test} \pdldia{\pra_{2}; \phi_{2} \test} \pdlbox{\negnnf{\psi}_{2} \test} \cdots \pdldia{\pra_{n}; \phi_{n} \test} \pdlbox{\negnnf{\psi}_{n} \test}x\label{eq:odd-form-4}
\end{eqnarray}

We now explicitly list the solutions to fixed-point equations depending on which level of a hierarchy they are associated with. To make the presentation of our solutions more compact, we introduce the following notation for a sequence of compositions of programs: $\bigodot_{i=1}^{n} \pra_{i} := \pra_{1} ; \pra_{2} ; \ldots ; \pra_{n}$. The solutions to fixed-point equations are as follows: for $n \in \mathbb{N}$, (1) if $\phi(x) \in \biset^{x}_{2n}$, then $\sol_{1}$ is a solution to $x \equiv \phi(x)$, (2) if $\phi(x) \in \biset^{x}_{2n+1}$, then $\sol_{2}$ is a solution to $x \equiv \phi(x)$, (3) if $\phi(x) \in \diset^{x}_{2n}$, then $\sol_{3}$ is a solution to $x \equiv \phi(x)$, and (4) if $\phi(x) \in \diset^{x}_{2n+1}$, then $\sol_{4}$ is a solution to $x \equiv \phi(x)$.
$$
\underbrace{\pdlbox{\big(\negnnf{\phi}_{1}\test ; \bigodot_{i=2}^{n} \pra_{i} ; \negnnf{\phi}_{i}\test \big)^{*}} \bigwedge_{j=1}^{n} \pdlbox{\negnnf{\phi}_{1}\test ; \bigodot_{k=2}^{j} \pra_{k} ; \negnnf{\phi}_{k}\test} \psi_{j}}_{= \ \sol_{1}}
\quad
\underbrace{\pdlbox{\big(\bigodot_{i=1}^{n} \pra_{i} ; \negnnf{\phi}_{i}\test \big)^{*}} \bigwedge_{j=1}^{n} \pdlbox{ \bigodot_{k=1}^{j} \pra_{k} ; \negnnf{\phi}_{k}\test} \psi_{j}}_{= \ \sol_{2}}
$$
$$
\underbrace{\pdldia{\big(\negnnf{\phi}_{1}\test ; \bigodot_{i=2}^{n} \pra_{i} ; \negnnf{\phi}_{i}\test \big)^{*}} \bigvee_{j=1}^{n} \pdldia{\negnnf{\phi}_{1}\test ; \bigodot_{k=2}^{j} \pra_{k} ; \negnnf{\phi}_{k}\test} \psi_{j}}_{= \ \sol_{3}}
\
\underbrace{\pdldia{\big(\bigodot_{i=1}^{n} \pra_{i} ; \negnnf{\phi}_{i}\test \big)^{*}} \bigvee_{j=1}^{n} \pdldia{ \bigodot_{k=1}^{j} \pra_{k} ; \negnnf{\phi}_{k}\test} \psi_{j}}_{= \ \sol_{4}}
$$
We now argue that the above formulae $\sol_{i}$ are indeed solutions to their respective fixed-point equations.

\begin{theorem}\label{thm:fixed-point-solution}
If $\phi(x) \in \biset^{x} \cup \diset^{x}$, then $x \equiv \phi(x)$ has a solution $\sol$.
\end{theorem}

\begin{proof} We prove that $\sol_{2}$ is a solution to $x \equiv \phi(x)$ for $\phi(x) \in \biset^{x}_{2n+1}$. The proofs that $\sol_{1}$, $\sol_{3}$, and $\sol_{4}$ are solutions to $x \equiv \phi(x)$ for $\phi(x) \in \biset^{x}_{2n}$, $\phi(x) \in \diset^{x}_{2n}$, and $\phi(x) \in \diset^{x}_{2n+1}$, respectively, are proven similarly.

Since $\phi(x) \in \biset^{x}_{2n+1}$, it has the shape of formula (\ref{eq:odd-form-1}) discussed above. Observe that (\ref{pf1:step2}) can be directly obtained from (\ref{pf1:step1}) by applying E4. To obtain (\ref{pf1:step3}) from (\ref{pf1:step2}), we use the equivalences E1 and E3 in the following way: first, we apply E1 to `decompose' the modalities $\pdlbox{ \bigodot_{k=1}^{j} \pra_{k} ; \negnnf{\phi}_{k}\test}$ for $2 \leq j \leq n$ and $\pdlbox{\bigodot_{i=1}^{n} \pra_{i} ; \negnnf{\phi}_{i}\test }$ in (\ref{pf1:step2}) to obtain $\pdlbox{ \pra_{1} ; \negnnf{\phi}_{1}\test} \pdlbox{ \bigodot_{k=2}^{j} \pra_{k} ; \negnnf{\phi}_{k}\test}$ for $2 \leq j \leq n$ and $\pdlbox{\pra_{1} ; \negnnf{\phi}_{1}\test } \pdlbox{\bigodot_{i=2}^{n} \pra_{i} ; \negnnf{\phi}_{i}\test }$, respectively. Second, we apply E3 to distribute the modalities $\pdlbox{ \pra_{1} ; \negnnf{\phi}_{1}\test}$ occurring in front of each conjunct out in front of the entire conjunction, giving formula (\ref{pf1:step3}). To obtain (\ref{pf1:step4}) from (\ref{pf1:step3}), we use the equivalence E2. Observe that $\chi$ has a similar shape to (\ref{pf1:step2}), but the range of the parameters $j$ and $i$ has decreased from $[1,n]$ to $[2,n]$. We now repeat the process of applying E1, E3, and then E2, in a similar manner as described above until we obtain the formula (\ref{pf1:step4.5}). Above, we discussed that a formula of shape (\ref{eq:odd-form-2}) is equivalent to a formula of shape (\ref{eq:odd-form-1}), and therefore, formula (\ref{pf1:step4.5}) is equivalent to $\phi(\sol_{2})$. 
\begin{eqnarray}
\sol_{2}  & = & \pdlbox{\big(\bigodot_{i=1}^{n} \pra_{i} ; \negnnf{\phi}_{i}\test \big)^{*}} \bigwedge_{j=1}^{n} \pdlbox{ \bigodot_{k=1}^{j} \pra_{k} ; \negnnf{\phi}_{k}\test} \psi_{j}\label{pf1:step1}\\
 & \equiv & \bigwedge_{j=1}^{n} \pdlbox{ \bigodot_{k=1}^{j} \pra_{k} ; \negnnf{\phi}_{k}\test} \psi_{j} \land \pdlbox{\bigodot_{i=1}^{n} \pra_{i} ; \negnnf{\phi}_{i}\test } \sol_{2}\label{pf1:step2}\\
 & \equiv & \pdlbox{\pra_{1} ; \negnnf{\phi}_{1}\test} \big(\psi_{1} \land \bigwedge_{j=2}^{n} \pdlbox{ \bigodot_{k=2}^{j} \pra_{k} ; \negnnf{\phi}_{k}\test} \psi_{j} \land \pdlbox{\big(\bigodot_{i=2}^{n} \pra_{i} ; \negnnf{\phi}_{i}\test \big)} \sol_{2} \big)\label{pf1:step3}\\
 & \equiv & \pdlbox{\pra_{1} ; \negnnf{\phi}_{1}\test} \pdldia{\psi_{1}\test} \big(\underbrace{\bigwedge_{j=2}^{n} \pdlbox{ \bigodot_{k=2}^{j} \pra_{k} ; \negnnf{\phi}_{k}\test} \psi_{j} \land \pdlbox{\big(\bigodot_{i=2}^{n} \pra_{i} ; \negnnf{\phi}_{i}\test \big)} \sol_{2}}_{\chi}\big)\label{pf1:step4}
  \end{eqnarray}
 \begin{eqnarray}
 & \equiv & \pdlbox{\pra_{1}; \negnnf{\phi}_{1} \test} \pdldia{\psi_{1} \test} \pdlbox{\pra_{2}; \negnnf{\phi}_{2} \test} \pdldia{\psi_{2} \test} \cdots \pdlbox{\pra_{n}; \negnnf{\phi}_{n} \test} \pdldia{\psi_{n} \test} \sol_{2}\label{pf1:step4.5}
\end{eqnarray}
This concludes the proof.\qed
\end{proof}

\begin{example} We give an example showing how to solve the particular fixed-point equation $x \equiv \phi(x)$ with $\phi(x) = p \land \pdlbox{a}(q \lor (r \land x))$. First, observe that $\phi(x) \equiv \bot \lor (p \land \pdlbox{a}(q \lor (r \land x)))$, meaning, $\bot \lor \phi(x) \in \biset_{2}^{x}$ and $x \equiv \phi(x)$ has a solution of the form $\sol_{1}$ since $\bot \lor \phi(x)$ is of the form $\phi_{1} \lor ( \psi_{1} \land \pdlbox{\pra_{2}}( \phi_{2} \lor (\psi_{2} \land x))$ with $n = 2$, $\phi_{1} = \bot$, $\psi_{1} = p$, $\pra_{2} = a$, $\phi_{2} = q$, and $\psi_{2} = r$.

\begin{align*}
\sol_{1} & =  \pdlbox{(\top \test ; a ; \negnnf{q}\test)^{*}} (\pdlbox{ \top \test} p \land \pdlbox{\top \test ; a ; \negnnf{q} \test} r)\\
& \equiv  (\pdlbox{ \top \test} p \land \pdlbox{\top \test ; a ; \negnnf{q} \test} r) \land \pdlbox{\top \test ; a ; \negnnf{q}\test} \sol_{1}\tag{by E4}\\
& \equiv  \pdlbox{ \top \test} (p \land \pdlbox{a ; \negnnf{q} \test} r \land \pdlbox{a ; \negnnf{q}\test} \sol_{1})\tag{by E1, E3}\\
& \equiv  p \land \pdlbox{a}\pdlbox{\negnnf{q} \test} r \land \pdlbox{a}\pdlbox{\negnnf{q}\test} \sol_{1}\tag{by E1, E5}\\
& \equiv  p \land \pdlbox{a}\pdlbox{\negnnf{q} \test} (r \land \sol_{1})\tag{by E3}\\
& \equiv  p \land \pdlbox{a} (q \lor (r \land \sol_{1}))\tag{by E7}\\
\end{align*}
The above proof demonstrates that $\sol_{1}$ is indeed a solution to $x \equiv \phi(x)$. We note that the above proof does not make use of E2 as in the proof of \thm~\ref{thm:fixed-point-solution} and additionally makes use of E5 to remove the redundant $\pdlbox{\top \test}$ modality. 
\end{example}

\section{Concluding Remarks}\label{sec:conclusion}

In this paper, we identified a non-trivial class of solvable $\pdl$ fixed-point equations obtained by taking the union $\biset^{x} \cup \diset^{x}$ of two dual hierarchies of increasing formulaic complexity. We provided explicit solutions to each fixed-point equation dependent upon which level of which hierarchy the equation was associated with. Although it is straightforward to verify that the formulae $\sol_{1}$--$\sol_{4}$ are solutions to their respective fixed-point equations, the novelty of the present work consisted in having identified the class $\biset^{x} \cup \diset^{x}$ and these solutions in the first place.

In future work, we aim to investigate larger classes of solvable $\pdl$ fixed-point equations. For example, observe that for any formula $\phi(x) \in \biset^{x} \cup \diset^{x}$ of form (\ref{eq:odd-form-1}) or (\ref{eq:odd-form-3}), the variable $x$ does not occur anywhere in its context, i.e., in any $\phi_{i}$, $\psi_{i}$, or $\pra_{i}$. Therefore, it would be interesting to look into solutions to equations where this restriction is lifted. As discussed in the introduction, we also aim to use solutions to fixed-point equations, accompanied by proof-theoretic techniques~\cite{AfsLei22,Sha14}, to investigate Craig interpolation for $\pdl$ and to extrapolate these techniques to other modal fixed-point logics.


\begin{credits}
\subsubsection{\ackname} 
Work supported by the European Research Council (ERC) Consolidator Grant 771779 
 (DeciGUT).

\subsubsection{\discintname}
The author has no competing interests to declare that are relevant to the content of this article.
\end{credits}
%
%
%
 \bibliographystyle{splncs04}
 \bibliography{bibliography.bib}
\end{document}